\documentclass[oneside,english]{amsart}
\usepackage[T1]{fontenc}
\usepackage[utf8]{inputenc}
\usepackage{amsthm}
\usepackage{amssymb}

\makeatletter
\numberwithin{equation}{section}
\numberwithin{figure}{section}
\theoremstyle{plain}
\newtheorem{thm}{\protect\theoremname}
  \theoremstyle{definition}
  \newtheorem{defn}[thm]{\protect\definitionname}

\makeatother

\usepackage{babel}
  \providecommand{\definitionname}{Definition}
\providecommand{\theoremname}{Theorem}

\begin{document}

\title{Topological approach to solve P versus NP}

\author{KOBAYASHI, Koji}

\maketitle

\section{Overview}

This paper talks about difference between P and NP by using topological
space that mean resolution principle. I pay attention to restrictions
of antecedent and consequent in resolution, and show what kind of
influence the restrictions have for difference of structure between
P and NP regarding relations of relation.

First, I show the restrictions of antecedent and consequent in resolution
principle. Antecedents connect each other, and consequent become a
linkage between these antecedents. And we can make consequent as antecedents
product by using some resolutions which have same joint variable.
We can determine these consequents reducible and irreducible.

Second, I introduce RCNF that mean topology of resolution principle
in CNF. RCNF is HornCNF and that variable values are presence of restrictions
of CNF formula clauses. RCNF is P-Complete.

Last, I introduce TCNF that have 3CNF's character which relate 2 variables
relations with 1 variable. I show CNF complexity by using CCNF that
combine some TCNF. TCNF is NP-Complete and product irreducible. I
introduce CCNF that connect TCNF like Moore graph. We cannot reduce
CCNF to RCNF with polynomial size. Therefore, TCNF is not in P.

\section{Preparation}

In this paper, I use CNF description like this.
\begin{defn}
\label{def: CNF} About $F\in CNF$, set of truth value assignment
that $F$ become true is $\left[F\right]$, $F$ become false is $\overline{\left[F\right]}$,
only clauses $c\in F$ become false is $\widehat{\left[c\right]}$.
Number of variables of $c$ is $\left|c\right|$, size of $F$ is
$\left|F\right|$. Number of truth value assignments of $\left[F\right]$,
$\overline{\left[F\right]}$ and $\widehat{\left[c\right]}$ is $\left|\left[F\right]\right|$,
$\left|\overline{\left[F\right]}\right|$ and $\left|\widehat{\left[c\right]}\right|$.

The composition of the clauses $c\in F$ may be denoted by a subscript.
That is, $c_{i\cdots\overline{j}\cdots}=\left(x_{i}\vee\cdots\overline{x_{j}}\vee\cdots\right)$.
The subscript of a capital letter shall be either positive or negative
of a variable. For examples, $c_{I},c_{\overline{I}}$ means $c_{I},c_{\overline{I}}\in\left\{ c_{i},\overline{c_{i}}\right\} ,c_{I}\neq c_{\overline{I}}$.

About resolution, I will use the term ``Joint Variable'' as variables
that positive and negative variable which are included in each antecedents
and not included in consequent, and ``Positive Antecedent'' as antecedent
that have positive joint variable, ``Negative Antecedent'' as antecedent
that have negative joint variable. We treat some resolution that have
same joint variable. Such case, positive antecedent, negative antecedent
and consequent become set of clauses.
\end{defn}
In this paper, I use metric space of fomula like this.
\begin{defn}
\label{def: Metric space of fomula} A metric space of fomula is the
metric space of the set of a truth value assignment becoming false
in a fomula.
\end{defn}

\section{Resolution}

We arrange structure of the resolution.
\begin{thm}
\label{thm: Relation of Antecedents} Antecedents of a resolution
connect each others.\end{thm}
\begin{proof}
I prove it using reduction to absurdity. We assume that some resolution
have 0 or over 2 joint variable. 

The case that resolution have 0 joint variable contradicts a condition
of the resolution clearly. The case that resolution have 2 joint variable
contradicts a condition of the resolution because $c_{IJp\cdots}\vee c_{\overline{I}\overline{J}q\cdots}\rightarrow c_{Jp\cdots\overline{J}q\cdots}=\top$.

Therefore, this theorem was shown than reduction to absurdity.\end{proof}
\begin{thm}
\label{thm: Relation between antecedent and consequent}In $c_{ip\cdots},c_{\overline{i}q\cdots}\in F\in CNF$
and resolution $c_{ip\cdots}\wedge c_{\overline{i}q\cdots}\rightarrow c_{p\cdots q\cdots}$,
$\overline{\left[c_{p\cdots q\cdots}\right]}$ become linkage of $\overline{\left[c_{ip\cdots}\right]},\overline{\left[c_{\overline{i}q\cdots}\right]}$
at the intersection of antecendent and consequent $\overline{\left[c_{ip\cdots q\cdots}\right]}\cap\overline{\left[c_{\overline{i}p\cdots q\cdots}\right]}$.\end{thm}
\begin{proof}
It is self-evident than a definition.\end{proof}
\begin{defn}
\label{def: Clauses product}I will use the term ``Clauses product''
like following.

$\left(c_{P\cdots}\wedge c_{Q\cdots}\right)\times\left(c_{R\cdots}\wedge c_{S\cdots}\right)=c_{P\cdots R\cdots}\wedge c_{Q\cdots R\cdots}\wedge c_{P\cdots S\cdots}\wedge c_{Q\cdots S\cdots}$\end{defn}
\begin{thm}
\label{thm: Constitution of the resolution} Consequent is the Clauses
product with positive antecedent and negative antecedent that remove
joint variables.

$\left(c_{ip\cdots}\wedge c_{iq\cdots}\wedge\cdots\right)\wedge\left(c_{\overline{i}r\cdots}\wedge c_{\overline{i}s\cdots}\wedge\cdots\right)$

$\rightarrow c_{p\cdots r\cdots}\wedge c_{p\cdots s\cdots}\wedge\cdots\wedge c_{q\cdots r\cdots}\wedge c_{q\cdots s\cdots}\wedge\cdots=\left(c_{p\cdots}\wedge c_{q\cdots}\wedge\cdots\right)\times\left(c_{r\cdots}\wedge c_{s\cdots}\wedge\cdots\right)$\end{thm}
\begin{proof}
It is self-evident than below\ref{def: Clauses product}.\end{proof}
\begin{defn}
\label{def: Direct sum of Clauses} About $X\in CNF$, I will use
the term ``Direct sum of Clauses'' like that $X=Y\times Z$, $Y,Z\in CNF$
and there is no injection $X\ni c\mapsto f\left(c\right)\in Y\cup Z$.
\end{defn}
Therefore, Direct sum of Clauses $Y\times Z$ is clauses product as
$X$ and $Y,Z$ representation does not include $X$. We can introduce
reducible and irreducible by using direct sum of clauses.
\begin{defn}
\label{def: Product Irreducible}I will use the term ``Product Reducible''
like that if direct sum of clauses $x=y\times z\in CNF$ exists in
$x\subset X\in CNF$. I will use the term ``Product Irreducible''
like that if $X$ is not product reducible.
\end{defn}

\section{RCNF}

I introduce topology of deduction system to formula. For simplification,
I treat topology as formula.
\begin{defn}
\label{def: RCNF} About $F\in CNF$, I will use the term ``DCNF(Deduction
CNF)'' as formula that variables value are presence of restrictions
of CNF formula clauses. Especially, I will use the term ``RCNF(Resolution
CNF)'' and ``$RCNF\left(F\right)$'' as DCNF that deduction system
is resolution principle. Clauses become variables and resolution become
clauses in $RCNF\left(F\right)$. Antecedent become negative variables
and consequent become positive variables. And furthermore, RCNF does
not include variable that correspond to empty clause.
\end{defn}
That is, if

$F\supset c_{ip\cdots}\wedge c_{\overline{i}q\cdots}$,

then

$RCNF\left(F\right)\supset\left(c_{ip\cdots}\right)\wedge\left(c_{\overline{i}q\cdots}\right)\wedge\left(\overline{c_{ip\cdots}}\vee\overline{c_{\overline{i}q\cdots}}\vee c_{p\cdots q\cdots}\right)$.

And RCNF does not include variable correspond to empty clause, therefore
sufficiency of $F$ accords with $RCNF\left(F\right)$. Resolution
consequent is 1 or less, therefore $RCNF\left(F\right)\in HornCNF$.

RCNF is P-Complete.
\begin{thm}
\label{thm: P-Completeness of RCNF} RCNF is P-Complete.\end{thm}
\begin{proof}
Clearly RCNF is HornCNF and $RCNF\in P$, I should show that we can
reduce HornCNF to RCNF in logarithm space.

To reduce HornCNF to RCNF, I show 2-step procedures.

First, I reduce HornCNF to at most 3 variables clauses HornCNF. We
can reduce by using same way to reduce CNF to 3CNF. That is, each
clauses change follows with new variables.

$c_{I\overline{j}\overline{k}\overline{l}\cdots}\rightarrow c_{I\overline{j}\overline{0}}\wedge c_{0\overline{k}\overline{1}}\wedge c_{1\overline{l}\overline{2}}\wedge\cdots$

We can execute this reduction with logarithm space, pointer to consequent,
pointer to variable, counter that show already used variables.

Second, I reduce this HornCNF to RCNF. We can reduce by adding resolution
formula for each clauses. We can reduce HornCNF with unit resolution,
therefore it is enough to keep SAT by using resolution formula that
variables of antecedent decreases. That is;

$c_{R}\rightarrow\left(c_{R}\right)\wedge\left(\overline{c_{R}}\vee\overline{c_{\overline{R}}}\right)$

$c_{P\overline{q}}\rightarrow\left(c_{P\overline{q}}\right)\wedge\left(c_{P}\vee\overline{c_{P\overline{q}}}\vee\overline{c_{q}}\right)\wedge\left(\overline{c_{P}}\vee\overline{c_{\overline{P}}}\right)$

$c_{I\overline{j}\overline{k}}\rightarrow\left(c_{I\overline{j}\overline{k}}\right)\wedge\left(c_{I\overline{k}}\vee\overline{c_{I\overline{j}\overline{k}}}\vee\overline{c_{j}}\right)\wedge\left(c_{I\overline{j}}\vee\overline{c_{I\overline{j}\overline{k}}}\vee\overline{c_{k}}\right)$

$\wedge\left(c_{I}\vee\overline{c_{I\overline{j}}}\vee\overline{c_{j}}\right)\wedge\left(c_{I}\vee\overline{c_{I\overline{k}}}\vee\overline{c_{k}}\right)\wedge\left(\overline{c_{I}}\vee\overline{c_{\overline{I}}}\right)$

We can execute this reduction with logarithm space, pointer to consequent,
pointer to variable.

Above two reduction, we can reduce HornCNF to RCNF. Both reductions
use only logarithm space, we can execute all reduction in logarithm
space.

Therefor, RCNF is P-Complete.
\end{proof}
RCNF have structure that correspond to resolution restrictions.
\begin{thm}
\label{thm: RCNF granularity} If $F\in MUC$ have $\widehat{\left[c\right]}$
that is product irreducible, $RCNF\left(F\right)$ have antecedent
variables which correspond to $\widehat{\left[c\right]}$ connected
component.\end{thm}
\begin{proof}
Mentioned above \ref{thm: Relation between antecedent and consequent},
$RCNF\left(F\right)$ have consequents that correspond to each clauses
linkages, and clauses that is resolved from another clauses also have
consequents. Therefore, to repeating resolution we can resolve the
consequents that correspond to $\widehat{\left[c\right]}$ linkages. 

From assumptions that $F\in MUC$, we must resolve these consequents
to empty clause. Therefore, $RCNF\left(F\right)$ have these consequents
as antecedents. 

Therefore, we can say the theorem is correct.
\end{proof}

\section{3CNF}

I show 3CNF complexity. I introduce TCNF that keep 3CNF complexity
and become product irreducible. And I show that we cannot reduce CCNF
that made by TCNF to RCNF with polynomial size.
\begin{thm}
\label{thm: Difference between 3CNF, HornCNF and 2CNF} 

3CNF associate 2 variables relations with other one variable.

HornCNF associate some true variables relations with other one true
variable.

2CNF associate one variable with other one variable.\end{thm}
\begin{proof}
About 3CNF, it is clearer than the next relations;

$\left(x_{P}\vee x_{Q}\vee x_{R}\right)\rightleftarrows\left(\overline{x_{P}}\wedge\overline{x_{Q}}\rightarrow x_{R}\right)\rightleftarrows\left(\overline{x_{P}}\rightarrow x_{Q}\vee x_{R}\right)$

About HornCNF, it is clearer than the next relations;

$\left(\overline{x_{p}}\vee\overline{x_{q}}\vee\cdots\vee x_{r}\right)\rightleftarrows\left(x_{p}\wedge x_{q}\rightarrow x_{r}\right)$

About 2CNF, it is clearer than the next relations;

$\left(x_{P}\vee x_{Q}\right)\rightleftarrows\left(\overline{x_{P}}\rightarrow x_{Q}\right)$
\end{proof}
Next, I introduce some CNF that we can treat easily with 3CNF's character
(association 2 variables relation with 1 variable).
\begin{defn}
\label{def: TCNF} I will use the term ``Triangular CNF'' and ``TCNF''
like that combine fomulas

$\left[T_{PQR}\right]=\overline{T_{PQR}}=c_{PQ\overline{R}}\wedge c_{P\overline{Q}R}\wedge c_{\overline{P}QR}$

therefore

$T_{PQR}=c_{\overline{P}\overline{Q}}\wedge c_{\overline{Q}\overline{R}}\wedge c_{\overline{P}\overline{R}}\wedge c_{PQR}$.

And I will use the term ``Triangular Minimal Unsatisfiable Core''
and ``TMUC'' like that is MUC of TCNF.

I show some TCNF's characters.\end{defn}
\begin{thm}
\label{thm: TCNF's NP-Completeness} TCNF is NP-Complete.\end{thm}
\begin{proof}
Clearly TCNF is 3CNF and $TCNF\in NP$, I should show that we can
reduce 3CNF to TCNF in polynomial size.

To reduce 3CNF to TCNF, I reduce $C_{\overline{X}Y\overline{Z}}\in3CNF$
to TCNF.

$T_{STX}\wedge T_{TUY}\wedge T_{UVZ}\in TCNF$ is

$\left(x_{X},x_{Y},x_{Z}\right)=\left(\top,\bot,\top\right)\leftrightarrow\forall x_{S},x_{T},x_{U},x_{V}\left(T_{STX}\wedge T_{TUY}\wedge T_{UVZ}=\bot\right)$,
then

$C_{\overline{X}Y\overline{Z}}\equiv T_{STX}\wedge T_{TUY}\wedge T_{UVZ}$.

And $T_{STX}\wedge T_{TUY}\wedge T_{UVZ}$ size is atmost polynomial
size of $C_{\overline{X}Y\overline{Z}}$, therefore TCNF that reduce
from 3CNF become atmost polynomial size.

Therefor, TCNF is NP-Complete.\end{proof}
\begin{thm}
\label{thm: TCNF Product Irreducible} About $T_{PQR}\in TCNF$, $T_{PQR}$
and $\left[T_{PQR}\right]$ is Product Irreducible.\end{thm}
\begin{proof}
I prove it using reduction to absurdity. We assume that $T_{PQR}$
or $\left[T_{PQR}\right]$ have Direct sum of Clauses $A\times B$.

About $A\times B$, $\left[a\right]$ is equivalence relation that
$a\in A$ is representative, $B_{\left[a\right]}$ is projection map
of $\left[a\right]$ to $B$.

We treat the relation between $B_{\left[a\right]}$ and $A\times B$.
If $\left|B_{\left[a_{p}\right]}\cap B_{\left[a_{q}\right]}\right|\leq1$
at all $a_{p},a_{q}\in A$, we can decide $t\in A\times B$ by using
$a_{p},a_{q}$ when $b\subset t$ is $b\in B_{\left[a_{p}\right]}\cap B_{\left[a_{q}\right]}$,
or by using $b$ when $b\notin B_{\left[a_{p}\right]}\cap B_{\left[a_{q}\right]}$.
And because mentioned above \ref{def: Clauses product}, $A\times B$
have clauses that correspond to $a\in A,\, b\in B$. And $A\times B$
have injection $A\times B\ni c\mapsto f\left(c\right)\in A\cup B$
and contradicts a condition of clauses product. Therefore $\left|B_{\left[a_{p}\right]}\cap B_{\left[a_{q}\right]}\right|>1$.

Calculate $\left|B_{\left[a_{p}\right]}\cap B_{\left[a_{q}\right]}\right|$
by using $T_{PQR}$ structure. The pair $\left(a,B_{\left[a\right]}\right)$
are;

$\left(c_{P},\left\{ c_{QR}\right\} \right)$,$\left(c_{\overline{P}},\left\{ c_{\overline{Q}},c_{\overline{R}}\right\} \right)$,$\left(c_{Q},\left\{ c_{PR}\right\} \right)$,$\left(c_{\overline{Q}},\left\{ c_{\overline{P}},c_{\overline{R}}\right\} \right)$,

$\left(c_{R},\left\{ c_{PQ}\right\} \right)$,$\left(c_{\overline{R}},\left\{ c_{\overline{P}},c_{\overline{Q}}\right\} \right)$,$\left(c_{PQ},\left\{ c_{R}\right\} \right)$,$\left(c_{QR},\left\{ c_{P}\right\} \right)$,$\left(c_{PR},\left\{ c_{Q}\right\} \right)$.

$\left|B_{\left[a_{p}\right]}\cap B_{\left[a_{q}\right]}\right|\leq1$
at all $a_{p},a_{q}\in A$ and contradicts a condition that $T_{PQR}=A\times B$.

Calculate $\left|B_{\left[a_{p}\right]}\cap B_{\left[a_{q}\right]}\right|$
by using $\left[T_{PQR}\right]$ structure. The pair $\left(a,B_{\left[a\right]}\right)$
are;

$\left(c_{P},\left\{ c_{Q\overline{R}},c_{\overline{Q}R}\right\} \right)$,$\left(c_{\overline{P}},\left\{ c_{QR}\right\} \right)$,

$\left(c_{Q},\left\{ c_{P\overline{R}},c_{\overline{P}R}\right\} \right)$,$\left(c_{\overline{Q}},\left\{ c_{PR}\right\} \right)$,

$\left(c_{R},\left\{ c_{P\overline{Q}},c_{\overline{P}Q}\right\} \right)$,$\left(c_{\overline{R}},\left\{ c_{PQ}\right\} \right)$,

$\left(c_{PQ},\left\{ c_{\overline{R}}\right\} \right)$,$\left(c_{\overline{P}Q},\left\{ c_{R}\right\} \right)$,$\left(c_{P\overline{Q}},\left\{ c_{R}\right\} \right)$,

$\left(c_{QR},\left\{ c_{\overline{P}}\right\} \right)$,$\left(c_{\overline{Q}R},\left\{ c_{P}\right\} \right)$,$\left(c_{Q\overline{R}},\left\{ c_{P}\right\} \right)$,

$\left(c_{PR},\left\{ c_{\overline{Q}}\right\} \right)$,$\left(c_{\overline{P}R},\left\{ c_{Q}\right\} \right)$,$\left(c_{P\overline{R}},\left\{ c_{Q}\right\} \right)$.

$\left|B_{\left[a_{p}\right]}\cap B_{\left[a_{q}\right]}\right|\leq1$
at all $a_{p},a_{q}\in A$ and contradicts a condition that $\left[T_{PQR}\right]=A\times B$.

Therefore, this theorem was shown than reduction to absurdity.
\end{proof}
We can make large CNF by combine some TCNF.
\begin{defn}
\label{def: CCNF} I will use the term ``Chaotic CNF'' and ``CCNF''
like that TCNF combine Moore graph(diameter k, degree 3) as $T\in TCNF$
are nodes and variables are edges, and $M_{k}\in CCNF$. I will use
the term ``Chaotic MUC'' and ``CMUC'' like that CCNF of MUC.

If we decide a $t_{0}\subset M_{k}$, I will use the term ``$t_{m}$''
like that distance from $t_{0}$become m. If $x_{P},x_{Q},x_{R}\in t_{m}$
and all $t_{m-1}$ are $x_{Q},x_{R}\notin t_{m}\ni x_{P}$, I will
use the term ``$t_{P-QR}$''.\end{defn}
\begin{thm}
\label{thm: Existence of CMUC} $M_{k}\in CMUC$ exists.\end{thm}
\begin{proof}
I prove it using mathematical induction.

Case $k=1$. 

$t_{0}=T_{PQR}$ 

and 

$M_{1}=T_{PQR}\wedge T_{P-S\overline{T}}\wedge T_{Q-T\overline{U}}\wedge T_{R-U\overline{S}}$ 

is MUC and CCNF. Therefore CMUC exists at $k=1$.

In $k=n$ then $M_{n}\in MUC$. 

$T_{M-PQ}\wedge T_{N-\overline{Q}R}=O_{n}\subset M_{n}$, $T_{M-PQ}\wedge T_{N-\overline{Q}R}$
is 

$X_{M}=X_{N}\rightarrow X_{P}\neq X_{R}$ and

$X_{M}\neq X_{N}\rightarrow X_{P}=X_{R}$.

To construct $M_{n+1}$ by using $M_{n}$, we rename $X_{Q}$ to new
variable $X_{S}$ and $X_{Q}$ to independ each other, and we add
$T_{n+1}$ with $X_{Q},X_{S}$ keeping same relations. For examples,

$O_{n,n+1}=T_{M-PQ}\wedge T_{N-RS}\wedge T_{P-UV}\wedge T_{R-VW}\wedge T_{\overline{Q}-XY}\wedge T_{\overline{S}-\overline{Y}Z}$
is 

$X_{M}=X_{N}\neq\top\rightarrow\left(X_{U}=X_{W}\right)\neq\left(X_{X}=X_{Z}\right)$
and

$X_{M}\neq X_{N}\rightarrow\left(X_{U}=X_{W}\right)=\left(X_{X}=X_{Z}\right)$.

In the same way, 

$T_{M-PQ}\wedge T_{N-QR}=A_{l}\subset M_{l}$, $T_{M-PQ}\wedge T_{N-QR}$
is 

$X_{M}=X_{N}\rightarrow X_{P}=X_{R}$ and

$X_{M}\neq X_{N}\rightarrow X_{P}\neq X_{R}$.

And we can construct $M_{n+1}$ same way. For examples,

$A_{n,n+1}=T_{M-PQ}\wedge T_{N-RS}\wedge T_{P-UV}\wedge T_{R-VW}\wedge T_{Q-XY}\wedge T_{S-YZ}$
is

$X_{M}=X_{N}\rightarrow\left(X_{U}=X_{W}\right)=\left(X_{X}=X_{Z}\right)$
and

$X_{M}\neq X_{N}\rightarrow\left(X_{U}=X_{W}\right)\neq\left(X_{X}=X_{Z}\right)$.

Now we think structure of $T_{n}\subset M_{n}$. From Moor Graph structure,
all $T_{n}$ adjoin each other and make one circle. $X_{P},X_{R}$
also become part of this circle. And 3 degree moore graph's girth
is even, therefore we can treat $T_{n}$ divide two set of arches
that is adjoin each other. In this circle, $\left(X_{M},X_{N}\right)$
relation decide $\left(X_{P},X_{R}\right)$ relation, and $X_{Q}$
that shares these adjoin $T$ decides these relation.

In the same way, we think structure of $T_{n+1}\subset M_{n+1}$. 

From Moor Graph structure, all $T_{n+1}$ adjoin each other and make
one circle. $X_{U},X_{W},X_{X},X_{Z}$ also become part of this circle.
And 3 degree moore graph's girth is even, therefore we can treat $T_{n+1}$
divide two set of arches that is adjoin each other. In this circle,
$\left(X_{M},X_{N}\right)$ relation decide $\left(X_{U},X_{W}\right),\left(X_{X},X_{Z}\right)$
relation, and $\left(X_{V},X_{Y}\right)$ that shares these adjoin
$T$ decides these relation.

Thus, we can make $M_{n+1}\in MUC$ by keeping $T_{n+1}$ relation
in $M_{n+1}$ to same as $T_{n}$ relation in $M_{n}$.

Therefore, this theorem was shown than mathematical induction.\end{proof}
\begin{thm}
\label{thm: Complexity of CMUC} $F\in CMUC$ exist that $\left|\widehat{\left[c\right]}\right|$
of all $c\in F$ is not polynomial size.\end{thm}
\begin{proof}
I proof by $T_{n}\subset M_{n}$ and $T_{n}\wedge T_{n+1}\subset M_{n+1}$
in mentioned above \ref{thm: Existence of CMUC} become 

$\dfrac{\left|\left[T_{n+1}\right]\right|}{\left|\left[T_{n}\right]\right|}>\dfrac{\left|M_{n+1}\right|}{\left|M_{n}\right|}\rightarrow\dfrac{\left|\left[T_{n+1}\right]\right|}{\left|\left[T_{n}\right]\right|}\times\dfrac{\left|M_{n}\right|}{\left|M_{n+1}\right|}>1$.

I pay attention to structure of $T_{n},T_{n+1}$. Mentioned above
\ref{thm: Existence of CMUC}, $O_{n,n+1}$ construct from $O_{n}$
have atmost one $O_{n+1}$, and another is $A_{n+1}$. If we extend
$M_{n}$ to $M_{n+1}$, $M_{n+1}$ have $O_{n+1}$ same number of
$O_{n}$ include in $M_{n}$. Therefore

$\dfrac{\left|\left[T_{n+1}\right]\right|}{\left|\left[T_{n}\right]\right|}\times\dfrac{\left|M_{n}\right|}{\left|M_{n+1}\right|}=\dfrac{\left|\left[A_{n,n+1}\right]\right|}{\left|\left[A_{n}\right]\right|}\times\dfrac{\left|A_{n}\right|}{\left|A_{n,n+1}\right|}\quad\left(as\quad n\gg0\right)$.

First, I calculate $\dfrac{\left|A_{n}\right|}{\left|A_{n,n+1}\right|}$.
Because the fomula $A_{n,n+1}$ and $A_{n}$, $\dfrac{\left|A_{n}\right|}{\left|A_{n,n+1}\right|}=\dfrac{1}{3}$.

Second, I calculate $\dfrac{\left|\left[A_{n,n+1}\right]\right|}{\left|\left[A_{n}\right]\right|}$. 

$\left[A_{n}\right]$ is 

$\left(X_{M},X_{N},X_{P},X_{R}\right)=\left(\bot,\bot,\bot,\bot\right),\left(\bot,\bot,\top,\top\right),\left(\bot,\top,\top,\bot\right),\left(\top,\bot,\bot,\top\right),\left(\top,\top,\bot,\bot\right)$

and $\left[A_{n,n+1}\right]$ is

$\left(X_{M},X_{N},X_{U},X_{W},X_{X},X_{Z}\right)$

$=\left(\bot,\bot,\bot,\bot,\bot,\bot\right),\left(\bot,\bot,\bot,\bot,\top,\top\right),\left(\bot,\bot,\top,\top,\bot,\bot\right),$

$\left(\bot,\bot,\bot,\top,\top,\bot\right),\left(\bot,\bot,\top,\bot,\bot,\top\right),$

$\left(\bot,\top,\bot,\bot,\bot,\top\right),\left(\bot,\top,\bot,\top,\bot,\bot\right),\left(\bot,\top,\bot,\top,\top,\top\right),\left(\bot,\top,\top,\top,\bot,\top\right),$

$\left(\top,\bot,\bot,\bot,\top,\bot\right),\left(\top,\bot,\top,\bot,\bot,\bot\right),\left(\top,\bot,\top,\bot,\top,\top\right),\left(\top,\bot,\top,\top,\top,\bot\right),$

$\left(\top,\top,\bot,\bot,\bot,\bot\right),\left(\top,\top,\bot,\bot,\top,\top\right),\left(\top,\top,\top,\top,\bot,\bot\right),\left(\top,\top,\top,\top,\top,\top\right)$.

Thus, $\dfrac{\left|\left[A_{n,n+1}\right]\right|}{\left|\left[A_{n}\right]\right|}=\dfrac{17}{5}$.

Therefore, 

$\dfrac{\left|\left[T_{n+1}\right]\right|}{\left|\left[T_{n}\right]\right|}\times\dfrac{\left|M_{n}\right|}{\left|M_{n+1}\right|}=\dfrac{\left|\left[A_{n,n+1}\right]\right|}{\left|\left[A_{n}\right]\right|}\times\dfrac{\left|A_{n}\right|}{\left|A_{n,n+1}\right|}=\dfrac{17}{5}\times\dfrac{1}{3}=\dfrac{17}{15}>1\quad\left(as\quad n\gg0\right)$.

And this relation holds in all $T$ because $T$ are symmetry in $M_{n}$
structure.

Therefore, this theorem was shown.\end{proof}
\begin{thm}
\label{thm: RCNF Size} $F\in CNF:O\left(RCNF\left(F\right)\right)>O\left(\left|F\right|^{m}\right)\quad m:constant$
exists.\end{thm}
\begin{proof}
Mentioned above \ref{thm: RCNF granularity}, if $\widehat{\left[c\right]}$
is product irreducible, $RCNF\left(F\right)$ have antecedent variables
which correspond to $\widehat{\left[c\right]}$ connected component.
And mentioned above \ref{thm: Complexity of CMUC}, $\left|\widehat{\left[c\right]}\right|$
exceed polynomial size of $\left|F\right|$. Therefore, to reduce
CNF to RCNF, we cannot include within polynomial size. \end{proof}
\begin{thm}
\label{thm: CNF and HornCNF}$CNF\nleq_{p}RCNF\equiv_{L}HornCNF$\end{thm}
\begin{proof}
Mentioned above \ref{thm: P-Completeness of RCNF}, RCNF is P-Complete.
But mentioned above \ref{thm: RCNF Size}, we cannot reduce CNF to
RCNF in polynomial size. Therefore, $CNF\nleq_{p}RCNF\equiv_{L}HornCNF$.\end{proof}

\end{document}